\RequirePackage[l2tabu, orthodox]{nag}

\documentclass[11pt,a4paper]{article}

\usepackage[english]{babel}
\usepackage[T1]{fontenc}
\usepackage[utf8]{inputenc}

\usepackage{microtype}
\usepackage{mathtools}
\usepackage{graphicx}
\usepackage{blindtext}
\usepackage[hidelinks]{hyperref}
\usepackage{pdfpages}

\usepackage{csquotes}
\usepackage{booktabs}
\usepackage{standalone}
\usepackage{amssymb}
\usepackage[titletoc,title]{appendix}

\usepackage[margin=1in]{geometry}
\usepackage[binary-units=true]{siunitx}
\usepackage{lmodern}
\usepackage{verbatim}
\usepackage{todonotes}
\usepackage{subcaption}
\usepackage{ntheorem}
\usepackage{paralist}

\usepackage{braket}
\usepackage{tipa}
\usepackage{textcomp}
\usepackage[affil-it]{authblk}

\usepackage{caption}
\usepackage{titling}

\usepackage[symbol]{footmisc}
\usepackage{makeidx}
\usepackage{helvet}
\usepackage[inline]{enumitem}
\usepackage[style=numeric-comp,url=false,backend=bibtex,%
firstinits=true,sorting=none,clearlang=true,%
maxnames=3,minnames=1,uniquename=false,%
maxbibnames=10
]{biblatex}

\DefineBibliographyStrings{english}{%
  urlseen = {accessed},
}

\DeclareCiteCommand{\fullcite}
  {\defcounter{maxnames}{99}%
    \usebibmacro{prenote}}
  {\usedriver
     {\DeclareNameAlias{sortname}{default}}
     {\thefield{entrytype}}}
  {\multicitedelim}
  {\usebibmacro{postnote}}

\graphicspath{{./}{images/}}

\newtheorem{theorem}{Theorem}
\newtheorem*{proof}{Proof}
\newtheorem{definition}{Definition}

\newtheorem{lemma}{Lemma}
\newtheorem{corollary}{Corollary}

\makeindex
\DeclarePairedDelimiter\floor{\lfloor}{\rfloor}

\newcommand{\cent}{\hbox{\textrm {\rlap/c}}}

\def\bitcoin{%
  \leavevmode
  \vtop{\offinterlineskip 
    \setbox0=\hbox{\textnormal Q}%
    \setbox2=\hbox to\wd0{\hfil\hskip-.03em
    \vrule height .4ex width .15ex\hskip .08em
    \vrule height .4ex width .15ex\hfil}
    \vbox{\copy2\box0}
\box2}}

\def\signed#1{{\leavevmode\unskip\nobreak\hfil\penalty50\hskip2em
  \hbox{}\nobreak\hfil(#1)%
  \parfillskip=0pt \finalhyphendemerits=0 \endgraf}}

  \newsavebox\mybox{}

\makeatletter
\def\@maketitle{%
  \newpage
  \null
  \vskip 2em%
  \begin{center}%
  \let \footnote \thanks
    {\Large\bfseries \@title \par}%
    \vskip 1.5em%
    {\normalsize
      \lineskip .5em%
      \begin{tabular}[t]{c}%
        \@author
      \end{tabular}\par}%
    \vskip 1em%
    {\normalsize \@date}%
  \end{center}%
  \par
  \vskip 1.5em}
\makeatother

\title{Quantum Bitcoin: An Anonymous and Distributed Currency Secured by the
No-Cloning Theorem of Quantum Mechanics}

\author{Jonathan Jogenfors%
    \thanks{Electronic address:
\href{mailto:jonathan.jogenfors@liu.se}{jonathan.jogenfors@liu.se}}}
\affil{Information Coding Group,\\Department of Electrical
Engineering,\\Linköping University, Sweden}
\begin{document}

\begin{titlingpage}
    \maketitle

    \begin{abstract}
        The digital currency Bitcoin has had remarkable growth since it was
        first proposed in 2008. Its distributed nature allows currency
        transactions without a central authority by using cryptographic methods
        and a data structure called the blockchain. In this paper we use the
        no-cloning theorem of quantum mechanics to introduce Quantum Bitcoin, a
        Bitcoin-like currency that runs on a quantum computer. We show that our
        construction of quantum shards and two blockchains allows untrusted
        peers to mint quantum money without risking the integrity of the
        currency. The Quantum Bitcoin protocol has several advantages over
        classical Bitcoin, including immediate local verification of
        transactions. This is a major improvement since we no longer need the
        computationally intensive and time-consuming method Bitcoin uses to
        record all transactions in the blockchain. Instead, Quantum Bitcoin only
        records newly minted currency which drastically reduces the footprint
        and increases efficiency. We present formal security proofs for
        counterfeiting resistance and show that a quantum bitcoin can be re-used
        a large number of times before wearing out -- just like ordinary coins
        and banknotes. Quantum Bitcoin is the first distributed quantum money
        system and we show that the lack of a paper trail implies full anonymity
        for the users. In addition, there are no transaction fees and the
        system can scale to any transaction volume.
    \end{abstract}
\end{titlingpage}

\section{Introduction}\label{sec:introduction}
Modern society relies on money to function. Trade and commerce is performed
using physical tokens (coins, banknotes) or electronically (credit cards, bank
transfers, securities). Recently, cryptographic currencies such as Bitcoin have
emerged as a new method to facilitate trade in a purely digital environment
without the need for a backing financial institution. Common to all functioning
currencies is demand together with a controlled supply. Traditional,
government-backed currencies mint currency according to rules decided by
politics while Bitcoin works according to pre-defined rules. The currencies are
then protected from counterfeiting, either by physical copy-protection in the
case of coins, banknotes and cashier's checks, or in Bitcoin by applying
cryptography. A detailed description of Bitcoin is given in Section\nobreakspace \ref {sec:bitcoin}.

The laws of quantum mechanics have given rise to interesting applications in
computer science, from the quadratic speedup of unstructured database search due
to \textcite{Grover1996} to the polynomial-time algorithm for integer
factorization by \textcite{Shor1994}. These \enquote{quantum} algorithms are
faster than their classical counterparts, showing that some computing problems
can be solved more efficiently if a classical computer is replaced by a quantum one. In
addition, quantum states are disturbed when measured, which has given rise to to
quantum cryptography protocols such as BB84~\cite{Bennett1984} and
E91~\cite{Ekert1991}, where the latter uses the quantum phenomena of
entanglement. See \textcite{Broadbent2015} for a recent survey of quantum
cryptography.

This begs the question: can quantum mechanics help us design new, improved money
systems? The answer is yes. As shown by \textcite{Wiesner1983}, the
\textbf{no-cloning theorem}~\cite{Wootters1982} provides an effective
basis for copy protection, although Wiesner's results predated the actual
theorem. See Section\nobreakspace \ref {sec:quantum-money} for a more detailed
history of quantum money.

This paper introduces Quantum Bitcoin, a new currency that combines the
copy-protection of quantum mechanics with Bitcoin to build a payment system with
several advantages over existing systems. We present some necessary background
in Sections\nobreakspace \ref {sec:bitcoin} and\nobreakspace  \ref {sec:quantum-money}, followed by the main contribution in
Section\nobreakspace \ref {sec:quantum-bitcoin}. Then, we list the numerous advantages of Quantum
Bitcoin in Section\nobreakspace \ref {sec:comparison} and conclude in Section\nobreakspace \ref {sec:conclusion}. Due to
space constraints, the security analysis has been moved to Appendix\nobreakspace \ref {sec:analysis}.

\section{The Bitcoin protocol}\label{sec:bitcoin}
The Bitcoin protocol was proposed in 2008 by \textcite{Nakamoto2008}. The true
identity behind that pseudonym still remains a mystery, but the concepts
introduced in the original whitepaper have proven themselves by giving rise to a
currency with a market cap exceeding 6.4 billion USD as of April 2016.

In order for a currency to function, there must be a finite amount in
circulation as well as a controlled supply of new currency. Traditional
currencies such as USD and EUR are controlled by a central organization, usually
called the central bank. Bitcoin instead uses cryptography to distribute this
task over a peer-to-peer network of users on the Internet.

Central to Bitcoin is the \textbf{blockchain}, which is a distributed ledger
that records all transactions of every user. Using the blockchain, a user can
compute his or her account balance by summing over all transactions to and from
that account. A transaction is initiated by the sending party by digitally
signing and then broadcasting a transaction message. The receiver of the
transaction sees the transaction message, but is advised to wait until third
parties, \textbf{miners} independently verify its validity. Otherwise, the
sender could perform double-spending, where the same unit of currency is
simultaneously and fraudulently sent to several receivers without them noticing.

A miner receives the broadcast transaction message and checks his or her local
copy of the blockchain to check the transaction against the miner's local
policy~\cite{Okupski2015}. Usually, this means that the sender of the
transaction must prove that he or she has knowledge of the private key
corresponding to the public key of the originating account by using a signature.
Also, the miner checks that the transferred bitcoin have not been spent. If the
transaction is valid, the miner wants to append it to the blockchain.

Appending new data to the blockchain is the critical part of the Bitcoin
protocol, and it requires authentication of the appended data. Without
authentication, a malicious miner could add invalid transactions to the
blockchain, thereby defrauding users. Traditional authentication methods cannot
be used for this purpose, as Bitcoin miners are only loosely organized,
anonymous and untrusted. Instead, \textcite{Nakamoto2008} uses a
\textbf{proof-of-work} puzzle, an idea introduced by \textcite{Back2002}. Here,
miners authenticate their verification by proving that they have spent computing
power, and therefore energy. This prevents the Sybil attack~\cite{Douceur2002},
in which an attacker can flood a hypothetical voting mechanism. Such an attack
becomes prohibitively expensive since each \enquote{vote} must be accompanied by
a proof of spent energy.

The essentials of a proof-of-work puzzle is as follows: The data $d$ is appended
to a random nonce value $r$ to produce $r+d$. This is fed to a hash function $f$
to produce the hash value $h=f(r+d)$. Next, the hash value is compared to a
certain \textbf{threshold}. If $h$ (interpreted as an integer on hexadecimal
form) is smaller than the threshold value, the transaction is verified and $d$
together with $r$ is then broadcast to the network. The nonce value $r$ can be
seen as a solution to the proof-of-work puzzle $d$. The solution is easily
verified, as it only requires one hashing operation $f$. If the nonce $r$ is not
a solution to the proof-of-work puzzle, the miner will have to try a new random
nonce $r$ and the process repeats. In fact, finding pre-images to secure hash
functions is computationally difficult and requires a large number of trials.

Bitcoin implements the proof-of-work puzzle by packing a number of transactions
into a so-called \textbf{block}. Each block contains, among other things, a
timestamp, the nonce, the hash value of the previous block, and the
transactions~\cite{Okupski2015}. The previous hash value fulfills an important
function, as it prevents the data in previous blocks being modified. This
imposes a chronological order of blocks, and the first Bitcoin block, called the
\textbf{Genesis block} was mined on January 3rd 2009.

Bitcoin miners are rewarded for their work by giving them newly minted bitcoin.
In fact, this is the only way in which new bitcoin are added to the network and
this rate must be controlled and predictable in order to prevent runaway
inflation. As more and more miners solve the proof-of-work puzzle, the faster
new blocks will be found, and new bitcoin will be at a runaway rate. The same
thing happens as computers become faster and more specialized. Bitcoin prevents
inflation by dynamically scaling the \textbf{difficulty} of the proof-of-work
puzzle to reach a target of one block found, on average, every ten
minutes~\cite{Okupski2015}. The difficulty is controlled via the threshold, or
the number of leading zeros required in the hash value. In
Appendix\nobreakspace \ref {sec:reuse-analysis} we designate $T_{block}$ as the average time between
blocks, so that Bitcoin uses $T_{block}=\SI{600}{\second}$. A quantitative study
by \textcite{Karame2012} suggests that the distribution of measured mining times
corresponds to a shifted geometric distribution with parameter 0.19.

The mining reward is implemented as a special type of transaction, called a
\textbf{coinbase}~\cite{Okupski2015} which is added to the block by the miner.
The reward size was originally 50 bitcoin, and is halved every 210000
blocks or approximately four years. This predictable reduction of mining reward is an
\textbf{inflation control scheme} since it controls the long-term supply of new
currency. It is expected that the last new bitcoin will be mined in the year
2140, when the reward falls below \num{e-8}, the smallest accepted bitcoin
denomination. There will only be 21 million bitcoin at this point,
however mining is expected to continue since miners also collect transaction
fees~\cite{Nakamoto2008,Kaskaloglu2014}.

When a new block is found it is added to the blockchain. All other miners must
then restart their progress, as the transactions they attempted to include have
already been included in a block. There is a possibility that a block has been
mined by a malicious miner, so the other miners will themselves check all
transactions in that block to see that they are valid. If a miner is satisfied
with the block and its hash value, it will restart the mining process based on the
newly mined block. If the block is invalid, it will be ignored by the network.

There is also a possibility of two miners independently mining a block, causing
a \textbf{fork}. There will then be ambiguity as to which block is considered
the valid one, and miners will randomly choose which block they choose as
starting point. With high probability, one of these branches will be the longest
one, causing the majority of miners to switch to that branch. Thus, the network
resolves forks by itself, at the cost of a nonzero probability that newly mined
blocks will be abandoned. Therefore, Bitcoin users are advised to wait until a
transaction has been \textbf{confirmed} by at least six consecutive
blocks~\cite{Karame2012}. Otherwise, there is a possibility that the block
containing the transaction is invalid, thereby invalidating the entire
transaction.

The goal of the blockchain is to prevent invalid transactions. To perform a
double-spend, an attacker must convince a user by mining at least six blocks.
A benevolent miner will not verify an invalid transaction, so the only way to
get it included in a block is for the attacker to mine it himself. This is done
in competition with the benevolent miners, so the probability of success depends
on the number of proof-of-work trials per second the attacker can perform. In
addition, the miner must win against the benevolent miners six times -- in a
row.

According to \textcite{Nakamoto2008} the probability for a malicious miner to
succeed in verifying an invalid transaction is exponentially small in the number
of confirmations as long as a majority of miners (i.e. computing power) is used
for benevolent purposes. This implies that the Bitcoin protocol is resistant to
double-spending attacks. However, each confirmation takes 10 minutes to finish,
so those six confirmations need one hour to finish, making transactions slow. In
addition, \textcite{Karame2012} found considerable variance in the time it takes
to mine a block; they measured a standard deviation of mining time of almost 15
minutes. Bitcoin users must therefore make a decision between security and
faster transaction times.

\section{Previous Proposals for Quantum Money}\label{sec:quantum-money}
As early as around 1970, \textcite{Wiesner1983} proposed a scheme that uses the
quantum mechanics to produce unforgeable \textbf{quantum
banknotes}~\cite{Broadbent2015}, however it took time for this result to be
published. The paper was initially rejected~\cite{Brassard2005} and according to
\textcite{Aaronson2012} it took 13 years until it was finally
publishedin 1983~\cite{Wiesner1983}. In the same year, BBBV~\cite{Bennett1983}
made improvements to Wiesner's scheme, such as an efficient way to keep track of
every banknote in circulation. Another, more recent, extension by
\textcite{Pastawski2012} increases the tolerance against noise. Even more
recently, \textcite{Brodutch2014} presented an attack on the Wiesner and BBBV
schemes.

After BBBV, quantum money received less attention due to the seminal 1984 paper
by \textcite{Bennett1984} that created the field of quantum key distribution
(QKD). Following two decades where virtually no work was done on quantum money,
\textcite{Mosca2006,Mosca2007,Mosca2009} proposed \textbf{quantum coins} around
ten years ago. In contrast to quantum banknotes (where each banknote is unique),
quantum coins are all identical.

We distinguish between \textbf{private key} and \textbf{public-key} quantum
money systems. In a private-key system, only the bank that minted the quantum
money can verify it as genuine, while a public-key system allows anyone to
perform this verification. The advantages of a public-key system over a
private-key one are obvious, assuming similar security levels. Until recently,
all quantum money proposals were private-key, however in 2009
\textcite{Aaronson2009} proposed the first public-key quantum money system.
While this system was broken in a short time by \textcite{Lutomirski2009}, it
inspired others to re-establish security. A novel proposal by
\textcite{Farhi2010} produced a public-key system using knot theory and
superpositions of link diagrams, and this idea was further developed by
\textcite{Lutomirski2011}. Finally, \textcite{Aaronson2012} based a public-key
quantum money scheme on the hidden-subspace problem.

Another important distinction is between systems that have unconditional
security, and those secure under computational hardness assumptions. In an
unconditionally secure quantum money scheme, no attacker can break the system
even when given unlimited computation time. For instance, Wiesner's scheme is
unconditionally secure while BBBV is not. According to \textcite{Farhi2010},
public-key quantum money cannot be unconditionally secure. Instead, the
proposals by \textcite{Aaronson2009, Farhi2010, Aaronson2012} all rely on
computational hardness assumptions, as will ours.

Common to all proposals discussed above is a centralized topology, with a number
of users and one \enquote{bank} that issues (and possibly verifies) money. This
requires all users to fully trust this bank, as a malevolent bank can perform
fraud and revoke existing currency.

\section{Quantum Bitcoin}\label{sec:quantum-bitcoin}
In this paper we present the inner workings of Quantum Bitcoin, a quantum
currency with no central authority. As with most quantum money schemes the
central idea is the no-cloning theorem~\cite{Wootters1982} which shows that it
is impossible to copy an arbitrary quantum state $\ket \psi$. Quantum mechanics
therefore provides an excellent basis on which to build a currency, as
copy-protection is \enquote{built in}. In Appendix\nobreakspace \ref {sec:analysis} we quantify the
level of security the no-cloning theorem gives, and show that our Quantum
Bitcoin are secure against counterfeiting. For brevity, we will refer to the
classical Bitcoin protocol simply as \enquote{Bitcoin} for the rest of the
paper.

\subsection{Prerequisites}
Quantum Bitcoin uses a classical blockchain, just like the Bitcoin
protocol. For the purposes of this paper, we model the blockchain as a
random-access ordered array with timestamped dictionary entries. Blocks can be
added to the end of the chain by solving a proof-of-work puzzle, and blocks in
the chain can be read using a lookup function. In the Quantum Bitcoin
blockchain, the blocks only contain descriptions of newly minted Quantum
Bitcoin. Transactions are not recorded as they are finalized locally. We can
therefore model each block as a dictionary data structure, where dictionaries
are key-value pairs that match serial numbers $s$ to public keys\footnote{Do not
    confuse the public key $k_{public}$ with the key of the dictionary}
    $k_{public}$. We will use the following formal definition:
\begin{definition}
    A classical distributed ledger scheme $\mathcal L$ consists of the following
    classical algorithms:
    \begin{itemize}
        \item $\mathsf{Append}_{\mathcal L}$ is an algorithm which takes
            $(s,k_{public})$ as input, where $s$ is a classical serial number
            and $k_{public}$ a classical public key. The algorithm fails if the
            serial number already exists in a block in the ledger. Otherwise, it
            begins to solve a proof-of-work puzzle by repeated trials of random
            nonce values. The algorithm passes if the puzzle is solved, at
            which time the ledger pair $(s,k_{public})$ is added as a new block.
        \item $\mathsf{Lookup}_{\mathcal L}$ is a polynomial-time algorithm that
            takes as input a serial number $s$ and outputs the corresponding
            public key $k_{public}$ if it is found in the ledger. Otherwise, the
            algorithm fails.
    \end{itemize}
    \label{def:ledger}
\end{definition}
While our formal definition is independent of the underlying block format and
security rules, we suggest adopting those used in Bitcoin. In addition,
$\mathsf{Append}_{\mathcal L}$ runs continuously until it passes -- if another
miner solves a proof-of-work puzzle it simply restarts the process transparently
to the caller.

Quantum Bitcoin also uses classical digital signatures. The scheme used by
Bitcoin is 256-bit ECDSA, but we will not commit to a specific
algorithm for Quantum Bitcoin. Instead, we use the following abstract model,
adapted from \textcite{Aaronson2012}:
\begin{definition}
    A classical public-key digital signature scheme $\mathcal D$ consists of
    three probabilistic polynomial-time classical algorithms:
    \begin{enumerate}
        \item $\mathsf{KeyGen}_{\mathcal D}$ which takes as input a security
            parameter $n$ and randomly generates a key pair
            $(k_{private},k_{public})$.
        \item $\mathsf{Sign}_{\mathcal D}$ which takes as input a private key
            $k_{private}$ and a message $M$ and generates a (possibly
            randomized) signature
            $\mathsf{Sign}_{\mathcal D}(k_{private},M)$.
        \item $\mathsf{Verify}_{\mathcal D}$, which takes as input $k_{public}$, a
            message $M$, and a claimed signature $\omega$, and either accepts or
            rejects.
    \end{enumerate}
\end{definition}
The key pair $(k_{private},k_{public})$ follows the usual conventions for
public-key cryptography: The private key is to be kept secret, and it should be
computationally infeasible for an attacker to derive $k_{private}$ from
$k_{public}$. In Quantum Bitcoin, $k_{private}$ is only used to mint new
currency and should therefore be discarded when this is completed.

\subsection{The Hidden Subspace Mini-Scheme}
Quantum Bitcoin uses a so-called \textbf{mini-scheme} model, inspired by
\textcite{Lutomirski2009,Farhi2010,Aaronson2012}. As will become clear in
Appendix\nobreakspace \ref {sec:analysis}, a mini-scheme setup allows for a simple way to prove the
security of Quantum Bitcoin. The mini-scheme can only mint and verify one single
Quantum Bitcoin, and in Sections\nobreakspace \ref {sec:construction} and\nobreakspace  \ref {sec:reuse} we extend this to a
full Quantum Bitcoin system using a blockchain. The explicit mini-scheme we adopt is an adaptation
of the Hidden Subspace mini-scheme system introduced by \textcite{Aaronson2012}.
In this scheme, Quantum Bitcoin states are on the form
\begin{equation}
    \ket{A}=\frac1{\sqrt{|A|}}\sum_{x\in A}\ket{x},
    \label{eqn:superposition}
\end{equation}
where $A$ is a subspace of $\mathbb F_2^n$. Here, $\mathbb F_2^n$ are bit
strings of length $n$ and the subspace $A$ is randomly generated from a set of
$n/2$ secret generators. In the same spirit, we define $A^\perp$ as the
orthogonal complement to $A$, that is, the set of $y\in \mathbb F_2^n$ such that
$x \cdot y=0$ for all $x\in \mathbb F_2^n$. Next, we define a membership oracle
$U_A$:
\begin{equation}
    U_A\ket{x}=
    \begin{cases}
        \enskip        -\ket{x} &\text{ if } x \in A\\
        \enskip \ket{x} &\text{ otherwise, }
    \end{cases}
\end{equation}
which is used to decide membership in $A$. We will later show how this oracle
can be explicitly implemented. Using $U_A$, we can implement a
projector $\mathbb P_A$ onto the set of basis states in $A$:
\begin{enumerate}
    \item Initiate a control qubit $\ket{0}$
    \item Apply $H$ to the control qubit
    \item Apply $U_A$ to $\ket x$ conditioned on the control qubit being in
        state $\ket 1$
    \item Measure the control qubit in the Hadamard basis $\{\ket +, \ket - \}$
    \item Postselect on getting the outcome $\ket -$
\end{enumerate}
Therefore, $\mathbb P_A$ operates on $\ket x$ in the following way:
\begin{equation}
    \mathbb P_A\ket x=\frac 1 {\sqrt 2}\left( \ket 0 \ket x + \ket 1U_A\ket x
    \right)=
    \begin{cases}
        \enskip        \ket{-} \ket x \text{ if } x \in A\\
        \enskip \ket{+}\ket x \text{ otherwise. }
    \end{cases}
\end{equation}
We define $U_{A^\perp}$ and $\mathbb P_{A^\perp}$ in a similar way as above,
except we instead operate on $A^\perp$. Together with the projectors $\mathbb
P_A$ and $\mathbb P_{A^\perp}$ we can create a unitary operator
\begin{equation}
    V_A=
    H_2^{\otimes n} \mathbb P_{A^\perp}
    H_2^{\otimes n} \mathbb P_{A},
\end{equation}
where $H_2^{\otimes n}$ denotes the quantum Fourier transform over $\mathbb
F_2^n$. We will use $V_A$ to verify Quantum Bitcoin states, where we interpret
$V_A\ket \psi=\ket A$ as passing and $V_A\ket \psi=0$ as failing. \Textcite[p.
28]{Aaronson2012} show that $V_A$ is a projector onto $A$, and that $V_A$
accepts an arbitrary state $\ket \psi$ with probability $|\braket{\psi|A}|^2$.
Formally, we can define the mini-scheme as follows:
\begin{definition}
    The Hidden Subspace mini-scheme $\mathcal M$ consist of two polynomial-time
    algorithms $\mathsf{Mint}_{\mathcal
    M}$ and $\mathsf{Verify}_{\mathcal M}$.
\end{definition}
Before we detail the minting and verification algorithms, we we need a way to
generate and verify quantum states and serial numbers. These algorithms should
have the following general structure:
\begin{definition}
    A \textbf{state generator} $\mathcal G(r)$ takes a random $n$-bit
    string $r$ and returns $(s_r, \langle A_r \rangle)$, where $s_r$ is a
    $3n$-bit string and $\langle A_r \rangle$ is a set of linearly independent
    generators $\{ x_1,\ldots,x_{n/2}\}$ for a subspace $A_r \leq \mathbb
    F_2^n$. We require that the serial numbers are distinct for every $r$.
    \label{def:quantum-state-generator}
\end{definition}
\begin{definition}
    A \textbf{serial number verifier} $\mathcal H(s)$ takes a serial number $s$
    and passes if it is a valid serial number $s=s_r$ for some $\langle A_r
    \rangle$ and fails otherwise.
    \label{def:serial-number-verifier}
\end{definition}
The algorithms for $\mathcal G(r)$ and $\mathcal H(s)$ can be implemented using
a random oracle, or explicitly using a scheme such as the multivariate
polynomial scheme introduced by \textcite[pp. 32--38]{Aaronson2012}. We are now
ready to complete the description of the mini-scheme $\mathcal M$ with the
following algorithms:
\begin{definition}
    $\mathsf{Mint}_{\mathcal M}(n)$ takes as input a security parameter $n$. It
    then randomly generates a secret $n$-bit key $r$ which it passes to the
    state generator $\mathcal G(r)$. The returned value is $(s_r,\langle
    A_r \rangle)$ which is used in Equation\nobreakspace \textup {(\ref {eqn:superposition})} to produce the
Quantum Bitcoin $(s_r,\rho)$, where $\rho=\ket{A}$.
\end{definition}

\begin{definition}
    $\mathsf{Verify}_{\mathcal M}(\cent{})$ takes as input an alleged Quantum
    Bitcoin \cent{} and performs the following checks, in order:
    \begin{enumerate}
        \item Form check: Accept if and only if $\cent{}$ has the form $(s,
            \rho)$, where $s$ is a classical serial number and $\rho$ is a
            quantum state. \item Serial number check: Accept if and only if the
                Serial Number
            Verifier $\mathcal H(s)$ accepts
        \item Apply $V_{A_r}=H_2^{\otimes 2}\mathbb P_{A_r^\perp}H_2^{\otimes
            2}\mathbb P_{A_r}$ to $\rho$ and accept if and only if
    $V_{A_r}(\rho)\neq 0$ \end{enumerate}
    Note that the verification procedure immediately fails if
    any of the above steps fail.
\end{definition}

\subsection{Naive Construction of Quantum Bitcoin}\label{sec:construction}
The mini-scheme $\mathcal M$ can only mint and verify one single quantum
bitcoin, so to build a usable Quantum Bitcoin ecosystem we need to extend the
model with a mechanism for minting and verifying any amount. For this purpose we
will define the full Quantum Bitcoin scheme, $\mathcal Q$, and implement it as
an extension of the mini-scheme $\mathcal M$. The connection between $\mathcal
M$ and $\mathcal Q$ is derived from the \enquote{standard construction} by
\textcite{Lutomirski2009,Farhi2010,Aaronson2012}. Formally, the definition of
the Quantum Bitcoin scheme is as follows:
\begin{definition}
    A public-key distributed Quantum Bitcoin scheme $\mathcal Q$ consists of the
    following algorithms:
    \begin{itemize}
        \item $\mathsf{KeyGen}_{\mathcal Q}$, a polynomial-time
            algorithm which takes as input a security
            parameter $n$ and randomly generates a key pair
            $(k_{private},k_{public})$.
        \item $\mathsf{Mint}_{\mathcal Q}$ which takes a security parameter $n$
            and a private key $k_{private}$ and generates a produces a quantum bitcoin \$.
        \item $\mathsf{Verify}_{\mathcal Q}$, a polynomial-time algorithm
            which takes as input an alleged quantum
            bitcoin $\cent{}$ and a corresponding public key $k_{public}$ and
            either accepts or rejects.
    \end{itemize}
    \label{def:quantum-money-scheme}
\end{definition}
Given a mini-scheme $\mathcal M=(\mathsf{Mint}_{\mathcal
M},\mathsf{Verify}_{\mathcal M})$, a digital signature scheme $\mathcal
D=(\mathsf{KeyGen}_{\mathcal D},\mathsf{Sign}_{\mathcal
D},\mathsf{Verify}_{\mathcal D})$ and a distributed ledger scheme $\mathcal
L=(\mathsf{Append}_{\mathcal L},\mathsf{Lookup}_{\mathcal L})$, we will
construct a first, intuitive, version of the Quantum Bitcoin scheme $\mathcal
Q=(\mathsf{KeyGen}_{\mathcal Q},\mathsf{Mint}_{\mathcal
Q},\mathsf{Verify}_{\mathcal Q})$. Later, we extend this standard construction
to protect against the reuse attack.

To begin the construction, we define $\mathsf{KeyGen}_{\mathcal Q}$ to simply be
$\mathsf{KeyGen}_{\mathcal D}$ from the digital signature scheme. Next, we
define the algorithm for
$\mathsf{Verify}_{\mathcal Q}$ for an alleged quantum bitcoin $\cent{}$:
\begin{enumerate}
    \item Check that $\cent{}$ is on the form
        $(s,\rho,\sigma)$, where the $s$ is a classical serial number, $\rho$ a
        quantum state, and $\sigma$ a classical digital
        signature.
    \item Call $\mathsf{Lookup}_{\mathcal L}(s)$ to retrieve
        the public key $k_{public}$ associated with
        the serial number $s$.
    \item Call $\mathsf{Verify}_{\mathcal
        D}(k_{public},s,\sigma)$ to
        verify the digital signature of the quantum bitcoin.
    \item Call $\mathsf{Verify}_{\mathcal M}(s,\rho)$ from the mini-scheme.
\end{enumerate}
The verification algorithm pass if and only if all of the above steps pass.
The main challenge of constructing Quantum Bitcoin is that the miners are
untrusted which is in contrast to previous quantum money schemes where minting
is done by a trusted entity such as a bank. In the same spirit as Bitcoin, the
intention is to take individually untrusted miners and still be able to trust
them as a group~\cite{Nakamoto2008}. Our first, Bitcoin-inspired attempt at the
$\mathsf{Mint}_{\mathcal Q}$ algorithm therefore becomes the following:
\begin{enumerate}
    \item Call $\mathsf{KeyGen}_{\mathsf Q}$ to randomly generate a key pair
        $(k_{private},k_{public})$.
    \item Generate a quantum bitcoin candidate by calling
        $\mathsf{Mint}_{\mathcal M}$, which returns $(s, \rho)$, where $s$ is a
        classical serial number and $\rho$ is a quantum state.
    \item Sign the serial number: $\sigma=\mathsf{Sign}_{\mathcal
        D}(k_{private},s)$.
    \item Call $\mathsf{Append}_{\mathcal L}(s,k_{public})$ to attempt to append
        the serial number $s$ and the public key $k_{public}$ to the ledger.
    \item If $\mathsf{Append}_{\mathcal L}$ failed, start again from step 2.
    \item If the serial number was successfully appended, put the serial number,
        quantum state and signature together to create the quantum bitcoin
        $\$=(s,\rho,\sigma)$.
\end{enumerate}
We can immediately identify the first major advantage of Quantum Bitcoin.
Whereas Bitcoin requires each transaction to be recorded into the blockchain --
a time-consuming process, Quantum Bitcoin transactions finalize immediately. Due
to the no-cloning theorem of quantum mechanics, the underlying quantum state in
the Quantum Bitcoin cannot be duplicated, thereby preventing counterfeiting in
itself (see Appendix\nobreakspace \ref {sec:analysis}). The only step of the protocol that uses
$\mathsf{Append}_{\mathcal L}$ is minting, which \enquote{normal} users don't
have to worry about.

\subsection{Preventing the Reuse Attack}\label{sec:reuse}
Our first attempts at the $\mathsf{Verify}_{\mathcal Q}$ and
$\mathsf{Mint}_{\mathcal Q}$ algorithms appear to work, but there is a problem.
In Quantum Bitcoin there is no trust assumption, so the users minting Quantum
Bitcoin can no longer be trusted to play by the rules as in the system by
\textcite{Aaronson2012}. This leads to a weakness that can be exploited using a
\textbf{reuse attack}: Ideally, the output of the minting algorithm and state
generator should be a unpredictable, even when fed the same argument twice.
Unfortunately, we must assume that all steps in $\mathsf{Mint}_{\mathcal Q}$ are
deterministic, with the obvious exception of $\mathcal{KeyGen}_{\mathcal Q}$.
Therefore, it is possible that a malicious miner generates a quantum bitcoin,
appends it to the blockchain, and then covertly reuse $k_{private}$ to produce
any number of quantum bitcoin that all pass verification.

This is a serious problem, since it allows that malicious miner to undermine the
payment system at any time. Imagine a scenario where a miner learns that the
quantum bitcoin he or she mined last year now is in possession by a political
opponent. The miner could then use $k_{private}$ to create a number of
identical, genuine, quantum bitcoin and disperse them everywhere. The natural
consequence is that the quantum bitcoin held by the opponent becomes worthless.
Compare this with Bitcoin. There, the blockchain records all transactions and a
miner therefore relinquishes control over the mined bitcoin as soon as it is
handed over to a recipient. In Quantum Bitcoin, however, there is no record of
who owns what, so there is no way to distinguish between the real and
counterfeit quantum bitcoin.

We prevent the reuse attack by adding a secondary stage to the minting
algorithm, where data is also appended to a new ledger $\mathcal L^{\prime}$. In
Appendix\nobreakspace \ref {sec:reuse-analysis} that this method makes the reuse attack improbable.
For the secondary mining step we introduce security parameters $m \geq 1$
and $T_{max}>0$ and the algorithm is as follows:
\begin{enumerate}
    \item A miner (this time called a \textbf{quantum shard miner}) uses the
        above \enquote{intuitive} minting scheme, but the finished product
        $(s,\rho,\sigma)$ is instead called a \textbf{quantum shard}.
    \item Quantum shard miners sell the quantum shards on a marketplace.
    \item Another miner (called a \textbf{quantum bitcoin miner}) purchases
        $m$ quantum shards $\{(s,\rho_i,\sigma_i)\}_{1\leq i \leq m}$ on the
        marketplace that, for all $1\leq i \leq m$, fulfill the following conditions:
        \begin{itemize}
            \item $\mathsf{Verify}_{\mathcal M}((s,\rho_i,\sigma_i))$ accepts
            \item The timestamp $T$ of the quantum shard in the Quantum Shard
                ledger $\mathcal L$ fulfills $t-T\leq T_{max}$, where $t$ is the
                current time.
        \end{itemize}
    \item The quantum bitcoin miner calls $\mathsf{KeyGen}_{\mathcal Q}$ to
        randomly generate a key pair $(k_{private},k_{public})$.
    \item The quantum bitcoin miner takes the serial numbers of the $m$ quantum
        shards and compiles the \textbf{classical descriptor}
        $s=(s_1,\ldots,s_m)$ and signs it as
        $\sigma_0=\mathsf{Sign}_{\mathcal D}(k_{private},s)$.
    \item The quantum bitcoin miner takes the $m$ quantum shards and, together
        with $\sigma_0$, produces a \textbf{quantum bitcoin candidate}:
        $(s_1,\rho_1,\sigma_1,\ldots,s_m,\rho_m,\sigma_m,\sigma_0)$.
    \item The quantum bitcoin miner calls $\mathsf{Append}_{\mathcal
        L^{\prime}}(s,k_{public})$ to
        attempt to pair the quantum bitcoin miner's public key $k_{public}$ with
        the classical descriptor $s$ in the ledger. Here, we
        require that $\mathsf{Append}$ fails if any of the $m$ quantum shards
        already have been combined into a quantum bitcoin that exists in the
        ledger $\mathcal L^{\prime}$.
\end{enumerate}
This process is the complete quantum mining protocol, and it works because each
participant is incentivized: quantum shard miners invest computing power to
produce quantum shards, which quantum bitcoin miners want for quantum bitcoin
production. As there is only a finite number of quantum shards, they will have
nonzero monetary value, thus rewarding the quantum shard miners for the energy
necessarily consumed. In turn, quantum bitcoin miners invest computing power to
mint quantum bitcoin from quantum shards. The quantum bitcoin miners are
rewarded with valid quantum bitcoin, which, due to their limited supply, are
expected to have nonzero monetary value.

According to \textcite{Nakamoto2008}, such incentives \enquote{may help nodes
to stay honest} and an attacker who has access to more computing power than the
rest of the network combined finds that it is more rewarding to play by the
rules than commit fraud. Also, the reuse attack is prevented
because two-stage mining makes it overwhelmingly difficult for a single entity
to first produce $m$ quantum shards and then combine them to a quantum
bitcoin.

Note the requirement that the quantum shards are less than $T_{max}$ old. This
is needed because the probability of an attacker successfully mining a quantum
shard approaches 1 as time goes to infinity. Therefore, given enough time, a
malicious miner can produce $(1-\varepsilon)m$ valid quantum shards which it
then could combine into a valid quantum bitcoin. The parameter $T_{max}$
prevents this from happening by expiring old quantum shards before this can
happen. In addition, while the algorithm makes use of \emph{two} ledgers instead
of one, it should be trivial to encode the two ledgers into one single
blockchain.

What remains is to slightly modify $\mathsf{Verify}_{\mathcal Q}$ to take
two-stage mining into account. We introduce an additional security parameter
$\lambda>0$ which will be determined later.
\begin{enumerate}
    \item Check that $\cent{}$ is on the form
        $(s_1,\rho_1,\sigma_1,\ldots,s_m,\rho_m,\sigma_m,
        \sigma_0)$, where the $s_i$ are classical serial numbers, $\rho_i$ are
        quantum states, and $\sigma_i$ (including $\sigma_0$) are digital
        signatures.
    \item Call $\mathsf{Lookup}_{\mathcal L^{\prime}}((s_1,\ldots,s_m))$ to
        retrieve the public key $k_{public}$ of the quantum bitcoin miner
        associated with the classical descriptor $(s_1,\ldots,s_m)$.
    \item Call $\mathsf{Verify}_{\mathcal
        D}(k_{public},(s_1,\ldots,s_m),\sigma_0)$ to
        verify the digital signature of the quantum bitcoin.
    \item For each $1\leq i \leq m$, call $\mathsf{Lookup}_{\mathcal L}(s_i)$ in
        order to retrieve the corresponding public keys $k_{public,i}$ from the
        quantum shard miners.
    \item For each $1 \leq i \leq m$, call $\mathsf{Verify}_{\mathcal
        D}(k_{public,i},s_i,\sigma_i)$ to verify the digital signatures of each
        of the quantum shards.
    \item For $1\leq i\leq m$, call $\mathsf{Verify}_{\mathcal M}(s_i,\rho_i)$
        to verify the quantum bitcoin states.
\end{enumerate}
The above algorithm checks the digital signatures of both the quantum bitcoin
and all contained quantum shards before calling the verification procedure of
the mini-scheme $\mathcal M$. The verification passes if and only if at least
$(1-\varepsilon-\lambda)m$ of the invocations of $\mathsf{Verify}_{\mathcal M}$
pass.

This concludes the description of Quantum Bitcoin. We have defined the minting
and verification algorithms and shown how it ties together with the blockchain
to build a working currency. What about security? In Appendix\nobreakspace \ref {sec:analysis} we will
give formal security proofs that
\begin{enumerate*}[label=(\roman*)]
    \item Quantum Bitcoin resists counterfeiting, or, more explicitly, that the
        false positive and false negative error probabilities of
        $\mathsf{Verify}_{\mathcal Q}$ are exponentially small in
        $n$,
    \item that the probability of a successful reuse attack is exponentially
        small in $1/T_{max}$, and
    \item that a quantum bitcoin can be used, or verified, an exponential
        number of times (in $n$) before wearing out.
\end{enumerate*}

\section{Comparison to Classical Bitcoin}\label{sec:comparison}
We will now compare Quantum Bitcoin to the classical Bitcoin protocol by
\textcite{Nakamoto2008} and show that Quantum Bitcoin has several advantages.
Bitcoin transactions must be verified by third-party miners. The time this takes
averages on one hour, but as previously mentioned the waiting time has
considerable variance~\cite{Karame2012}. Bitcoin transactions are therefore
slow; too slow for ma customer to wait in the check-out line. In contrast,
Quantum Bitcoin transactions are immediate and only requires the receiver to
have read-only access to a reasonably recent copy of the blockchain. We also
note that Quantum Bitcoin transactions are local, so that no blockchain must be
updated, nor does it require a third party to know of the transaction.

These local transactions are independent of network access. Bitcoin requires
two-way communication with the Internet, while Quantum Bitcoin transactions can
be performed in remote areas, including in space. The read-only blockchain
access requirement makes it possible to store a local off-line blockchain copy
in, for example, a book. To receive Quantum Bitcoin transaction, the user simply
needs to read from this book, given that the quantum bitcoin to be verified is
older than the book in question.

Another performance advantage is scalability. According to
\textcite{Garzik2015}, Bitcoin as originally proposed by \textcite{Nakamoto2008}
has an estimated global limit of seven transactions per second. In comparison,
the local transactions of Quantum Bitcoin implies that there is no upper limit
to the transaction rate. It should be noted, however, that the minting rate is
limited by the capacity of the Quantum Shard and Quantum Bitcoin blockchains. By
placing the performance restriction only in the minting procedure, the
bottleneck should be much less noticeable than if it were in the transaction
rate as well.

Local transactions also mean anonymity, since only the sender and receiver are
aware of the transaction even occurring. No record, and therefore no paper
trail, is created. In essence, a Quantum Bitcoin transaction is similar to that
of ordinary banknotes and coins, except no central point of authority has to be
trusted. Bitcoin, on the other hand, records all transactions in the blockchain
which allows anybody with a copy to trace transaction flows, even well after the
fact. This has been used by several
authors~\cite{Reid2013,Meiklejohn2013,Moser2013,Venkatakrishnan2013,Kondor2014,Androulaki2013}
to de-anonymize Bitcoin users.

Another advantage of Quantum Bitcoin is that transactions are free. Classical
Bitcoin transactions usually require a small fee~\cite{Nakamoto2008} to be paid
to miners in order to prevent transaction spam and provide additional incentives
for miners. It is also believed that these fees will allow mining to continue
past the year 2140, when the last new bitcoin is expected to be mined. In
Bitcoin, mining is required for transactions to work, but this is not the case
in Quantum Bitcoin, again due to local transactions. Even better, if Quantum
Bitcoin adopts an inflation control scheme similar to that of Bitcoin, there
will be no need for Quantum Bitcoin mining when the 21 million quantum bitcoin
have been found. Users will still be able to perform transactions even though
mining has stopped. When this occurs, it is realistic to publish a
\enquote{final version} of the Quantum Bitcoin blockchain in a book which then
would contain descriptors of all quantum bitcoin that will ever exist.

Compared to Bitcoin, the blockchain of Quantum Bitcoin is smaller and grows at a
more predictable rate. By nature, data added to a blockchain can never be
removed, and as of March 2016 the size of the Bitcoin blockchain exceeds
\SI{60}{\giga\byte}. This large size has made it difficult to implement a
complete Bitcoin implementation on smaller devices. Quantum Bitcoin also has a
growing blockchain, however it only grows when minting currency, not due to
transactions.

Per the discussion in the previous paragraph, Quantum Bitcoin mining could
become superfluous when all quantum bitcoin have been mined, which leads to an
upper limit of the Quantum Bitcoin blockchain. For example, if we limit the
number of quantum bitcoin to 21 million (as above) and choose 512-bit serial
numbers $s$ and a 256-bit digital signature scheme $\mathcal D$, the Quantum
Bitcoin blockchain will only ever grow to roughly \SI{2}{\giga\byte} in size
plus some overhead. This is an order of magnitude smaller than the
Bitcoin blockchain today, and much more manageable.

\section{Conclusion and Future Work}\label{sec:conclusion}
Quantum Bitcoin is a tangible application of quantum mechanics where we
construct the ideal distributed, publicly-verifiable payment system. The
currency works on its own without a central authority, and can start to work as
soon as it is experimentally possible to prepare, store, measure and reconstruct
quantum states with low enough noise. The no-cloning theorem provides
the foundation of copy-protection, and the addition of a blockchain
allows us to produce currency in a distributed and democratic fashion. Quantum
Bitcoin is the first example of a secure, distributed payment system with local
transactions and can provide the basis for a new paradigm for money, just like
Bitcoin did in 2008.

Two parties can transfer quantum bitcoin by transferring the Quantum Bitcoin
state over a suitable channel and reading off a publicly-available blockchain.
Transactions are settled immediately without having to wait for confirmation
from miners, and the Quantum Bitcoin can be used and re-constructed an
exponential number of times before they wear out. There is no transaction fee,
yet the system can scale to allow an unlimited transaction rate.

Note that while Quantum Bitcoin is secure against any counterfeiter with access
to a quantum computer, the protocol is not unconditionally secure. The
corresponding security proofs must therefore place the standard complexity
assumptions on the attacker. See Appendix\nobreakspace \ref {sec:analysis} for the complete security
analysis.

We invite further study of our proposal and welcome attempts at attacking this
novel protocol. There are also some challenges that should be addressed in
future work: quantum bitcoin are atomic and there is currently no way to
subdivide quantum bitcoin into smaller denominations, or merge them into larger
ones. A practical payment system would benefit greatly from such mechanisms as
it otherwise becomes impossible to give change in a transaction. The security
proofs in this paper should also be extended to the non-ideal case with the
addition of noise, decoherence and other experimental effects.

\section*{Acknowledgements}
The author would like to thank Niklas Johansson and Prof. Jan-Åke Larsson for
interesting discussions, feedback and proof-reading the manuscript.

\appendix
\begin{appendices}
    \section{Security Analysis}\label{sec:analysis}
    In this section we perform the security analysis of Quantum Bitcoin and show
    that it is secure against counterfeiting. Here, we reap the benefits of the
    mini-scheme setup as the proof becomes relatively easy. We begin by
    quantifying the probability of false negatives and false positives in the
    verification process and then we show that the mini-scheme is secure,
    followed by the observation that a secure mini-scheme implies security of
    the full system $\mathcal Q$.

    \subsection{Counterfeiting}\label{sec:counterfeiting}
    Our formal security analysis begins by modeling a counterfeiter, which is a
    quantum circuit that produces new, valid, quantum bitcoin outside of the normal
    minting procedure.
    \begin{definition}
        A \textbf{counterfeiter} $C$ is a quantum circuit of polynomial size (in
        $n$) which maps a polynomial (in $n$) number of valid quantum bitcoin to a
        polynomial number (in $n$) of new, possibly entangled alleged quantum
        bitcoin.
        \label{def:counterfeiter}
    \end{definition}
    A more detailed description of the \enquote{composite} counterfeiter of
    Quantum Bitcoin is given in \textcite[pp. 42--43]{Aaronson2012}. We next
    need to quantify the probability of a counterfeit quantum bitcoin to be
    accepted by the verification procedure. This is the probability of a false
    positive:
    \begin{definition}
        A Quantum Bitcoin scheme $\mathcal Q$ has \textbf{soundness error}
        $\delta$ if, given any counterfeiter $C$ and a collection of $q$ valid
        quantum bitcoin
        $\$_1,\ldots,\$_q$ we have
        \begin{equation}
            Pr(\mathsf{Count}(C(\$_1,\ldots,\$_q))>q)\leq \delta,
        \end{equation}
        where $\mathsf{Count}$ is a \textbf{counter} that takes as input a
        collection of (possibly-entangled) alleged quantum bitcoin $\$_1,\ldots
        ,\$_r$ and outpts the number of indices $0\leq i \leq r$ such that
        $\mathsf{Verify}(\$_i)$ accepts
        \label{def:soundness}
    \end{definition}
    Conversely, we quantify the probability of false negative, i.e. the probability
    that a valid quantum bitcoin is rejected by the verification procedure:
    \begin{definition}
        A Quantum Bitcoin scheme $\mathcal Q$ has \textbf{completeness
        error} $\varepsilon$ if $\mathsf{Verify}(\$)$ accepts with probability at least
        $1-\varepsilon$ for all valid quantum bitcoin $\$$. If $\varepsilon=0$ then
        $\mathcal Q$ has \textbf{perfect completeness}.
        \label{def:completeness}
    \end{definition}
    Next, we continue with analyzing the mini-scheme. Recall that
    a mini-scheme only mints and verifies one single quantum bitcoin, so that a
    mini-scheme counterfeiter only takes the single valid quantum bitcoin as
    input. To perform this analysis, we need a technical tool, the double
    verifier~\cite{Aaronson2012}:
    \begin{definition}
        For a mini-scheme $\mathcal M$, we define the \textbf{double verifier}
        $\mathsf{Verify_2}$ as a polynomial-time algorithm that takes as input a single
        serial number $s$ and two (possibly-entangled) quantum states $\rho_1$ and
        $\rho_2$ and accepts if and only if
        $\mathsf{Verify}_{\mathcal M}(s,\sigma_1)$ and $\mathsf{Verify}_{\mathcal
        M}(s,\sigma_2)$ both accept.
    \end{definition}
    Now, we define the soundness and completeness error for the mini-scheme:
    \begin{definition}
        A mini-scheme $\mathcal M$ has \textbf{soundness error} $\delta$ if, given
        any counterfeiter $C$,
        $\mathsf{Verify_2}(s,C(\$))$ accepts with probability at most $\delta$. Here
        the probability is over the quantum bitcoin or quantum shard $\$$ output by $\mathsf{Mint}_{\mathcal
        M}$ as well as the behavior of $\mathsf{Verify_2}$ and $C$.
    \end{definition}
    \begin{definition}
        A mini-scheme $\mathcal M$ has \textbf{completeness error} $\varepsilon$
        if $\mathsf{Verify}(\$)$ accepts with probability at least
        $1-\varepsilon$ for all valid quantum bitcoin or quantum shards $\$$. If
        $\varepsilon=0$ then $\mathcal Q$ has \textbf{perfect completeness}.
    \end{definition}
    We call a system \textbf{secure} if it has completeness error $\varepsilon
    \leq 1/3$ and soundness error $\delta$ exponentially small in $n$. While
    $1/3$ sounds like a high error probability, \textcite[pp.
    42--43]{Aaronson2012} show that the completeness error $\varepsilon$ of a
    secure system can be made exponentially small in $n$ at the small cost of
    increasing the soundness error $\delta$ from 0 to be exponentially small in
    $n$.

    What remains is to show that the Quantum Bitcoin system $\mathcal Q$ is, in
    fact, secure. This would be difficult had we not used the mini-scheme model, but
    now we can do this in a single step. The following
    theorem is adapted from \textcite[p. 20]{Aaronson2012}:
    \begin{theorem}[Security From The Mini-Scheme]
        If there exists a secure mini-scheme $\mathcal M$, then there also exists a
        secure Quantum Bitcoin scheme $\mathcal Q$. In particular, the
        completeness and soundness errors of $\mathcal Q$ are exponentially small
        in $n$.
        \label{thm:mini-scheme-reduction}
    \end{theorem}
    In order to prove Theorem\nobreakspace \ref {thm:mini-scheme-reduction}, we need the following
    lemma, adapted from \textcite[pp. 42--43]{Aaronson2012}:
    \begin{lemma}\label{lem:composition}
        Let $\mathcal{Q}^{\prime}=\left( \mathsf{KeyGen}^{\prime}_{\mathcal Q},
            \mathsf{Mint}^{\prime}_{\mathcal
        Q},\mathsf{Verify}^{\prime}_{\mathcal Q}\right)$ be the
        \enquote{naive} Quantum Bitcoin scheme described in
        Section\nobreakspace \ref {sec:construction}. If that scheme has completeness error
        $\varepsilon^{\prime}<1/2$ and soundness error $\delta^{\prime}
        <1-2\varepsilon^{\prime}$,
        then, for all polynomials $p$ and any
        $\delta>\frac{\delta^{\prime}}{1-2\varepsilon^{\prime}}$, the \enquote{composite}
        Quantum Bitcoin scheme $\mathcal{Q}=\left(
            \mathsf{KeyGen}_{\mathcal Q},
            \mathsf{Mint}_{\mathcal
        Q},\mathsf{Verify}_{\mathcal Q}\right)$ in
        Section\nobreakspace \ref {sec:reuse} has completeness error
        $\varepsilon=1/2^{p\left(  n\right)  }$\ and
        soundness error $\delta$.
    \end{lemma}
    We are now ready to prove Theorem\nobreakspace \ref {thm:mini-scheme-reduction}:
    \begin{proof}
        We use the quantum state generator $\mathcal G(r)$ from
        Definition\nobreakspace \ref {def:quantum-state-generator} as a one-way
        function: given a $n$-bit string $r$, $\mathcal G(r)$ outputs (among
        others) an unique $3n$-bit serial number $s_r$. If there exists a
        polynomial-time quantum algorithm to recover $r$ from $s_r$ it would be
        possible for a counterfeiter to copy quantum bitcoin, which is a contradicts
        the security of the mini-scheme. Therefore, $\mathcal G(r)$ is a one-way
        function secure against quantum attack. Since such one-way functions are
        necessary and sufficient for secure digital signature
        schemes~\cite{Rompel1990}, we immediately get a digital signature scheme
        $\mathcal D$ secure against quantum chosen-plaintext attacks. Next, we
        show that $\mathcal M$ and $\mathcal D$ together produce a secure
        \enquote{naive} Quantum Bitcoin system $\mathcal Q^{\prime}$, which is
        done in \textcite[p. 20]{Aaronson2012}.

        Finally, we choose $p(n)=n$ and use lemma\nobreakspace \ref {lem:composition} to show that
        the \enquote{composite} Quantum Bitcoin system described in
        Section\nobreakspace \ref {sec:reuse} has completeness error $\varepsilon=2^{-n}$ and any
        soundness error $\delta$ such that $\delta>3\delta^{\prime}$. Thus, we
        can make $\mathcal Q$ have soundness and completeness errors
        exponentially small in $n$. Therefore, the Quantum Bitcoin system
        $\mathcal Q$ is secure.
    \end{proof}
    This is the elegance of the mini-scheme model, where a secure mini-scheme immediately
    gives us the full, secure system. Note that in the proof of
    lemma\nobreakspace \ref {lem:composition} the security parameter $\lambda>0$ is taken to be
    sufficiently small. A counterfeiter who wants to break the security of
    $\mathcal Q$ is forced to break the security of $\mathcal M$, so if we now
    show that $\mathcal M$ is indeed secure we are finished:
    \begin{theorem}[Security Reduction for the Hidden-Subspace Mini-Scheme]
        The mini-scheme $\mathcal M=(\mathsf{Mint}_{\mathcal
        M},\mathsf{Verify}_{\mathcal M})$, which is defined relative to the classical
        oracle $U$, has zero completeness and soundness error exponentially
        small in $n$.
        \label{thm:mini-scheme-security}
    \end{theorem}
    \begin{proof}
        The Inner-Product Adversary Method by \textcite[p. 31]{Aaronson2012}
        gives an upper bound to the information gained by a single oracle query.
        Theorem\nobreakspace \ref {thm:mini-scheme-reduction} then shows that the mini-scheme
        $\mathcal M$ is secure since a valid quantum shard always passes
        verification (zero completeness error), and counterfeit quantum shard
        pass with only an exponentially small probability (soundness error is
        exponentially small in $n$).
    \end{proof}
    We can now state the main result:
    \begin{corollary}[Counterfeiting Resistance of Quantum Bitcoin]\label{thm:quantum-bitcoin}
        The Quantum Bitcoin system $\mathcal Q$ is secure. In particular, the
        completeness and soundness errors of $\mathcal Q$ are exponentially small
        in $n$.
    \end{corollary}
    \begin{proof}
        Theorem\nobreakspace \ref {thm:mini-scheme-security} shows that $\mathcal M$ is secure, and
        from Theorem\nobreakspace \ref {thm:mini-scheme-reduction} it then follows that $\mathcal Q$
        is secure. Explicitly, any counterfeiter must make
    $\Omega\left((1-\varepsilon)m2^{n/4}\right)$
        queries to successfully copy a single quantum bitcoin. For large enough $n$,
        this is computationally infeasible. In particular,
        Theorem\nobreakspace \ref {thm:mini-scheme-reduction} shows that the completeness and
        soundness errors of $\mathcal Q$ are exponentially small in $n$.
    \end{proof}
    Note that Quantum Bitcoin is not unconditionally secure. Therefore, it is
    conjectured that a hypothetical attacker without access to an exponentially
    fast computer cannot perform the exponential number of queries required to
    perform counterfeiting. Note that, according to \textcite{Farhi2010},
    public-key quantum money cannot be unconditionally secure, so this should
    not come as a surprise.

    \subsection{The Reuse Attack}\label{sec:reuse-analysis}
    Now we analyze the effect of the security parameters $m$ and $T_{max}$ on
    the probability of a reuse attack. A reuse attack is when the same entity
    first mines a number of quantum shards, then combines them into a quantum
    bitcoin. The security parameter $m$ controls the number of quantum shards
    required per quantum bitcoin, and $T_{max}$ is the maximum age of the
    quantum shards.

    For an attacker to perform the reuse attack, he or she must therefore mine
    $(1-\varepsilon)m-1=m-1-\varepsilon m$ quantum shards in $T_{max}$ seconds after a first quantum shard has
    been mined. Recall from Section\nobreakspace \ref {sec:reuse} that quantum shards expire after
    $T_{max}$ seconds. In reality the attacker must both mint quantum shards and
    combine them into quantum bitcoin before $T_{max}$ runs out. We simplify the
    analysis, however, by making it easier for the attacker and allow $T_{max}$
    time to mine quantum shards, and then again $T_{max}$ time to mine quantum
    bitcoin. We define $k:=\floor*{T_{max}/T_{block}}\geq 3$ as the average
    number of blocks mined before $T_{max}$ runs out, where $T_{block}$ is the
    average time between mined blocks. We have the following theorem:
    \begin{theorem}[Probability of Reuse Attack]\label{thm:reuse}
        The success probability of reuse attack in a secure Quantum Bitcoin system
        $\mathcal Q$ with completeness error $\varepsilon$ is exponentially small in
        $1/T_{max}$ as long as the attacker controls less than \SI{15}{\percent}
        of the computing power and $m < 1/\varepsilon$.
    \end{theorem}
    \begin{proof}
        We model the reuse attack by assigning $p$ to the probability of an
        attacker mining the next block in either the quantum shard or Quantum
        Bitcoin blockchain. $p$ can be understood as the proportion of the
        world's computing power controlled by the attacker.
        $\varepsilon$ is exponentially small in $n$, so
        with $n$ sufficiently large, the completeness error $\varepsilon$ is much
        smaller than $1/m$. The probability of the attacker mining
        $m-1$ of these $k$ quantum shards is then
        \begin{equation}
            \eta_1=    \binom k {m-1}
            p^{m-1}
            (1-p)^{k-m+1}.
        \end{equation}
        However, due to the way the verification algorithm
        $\mathsf{Verify}_{\mathcal M}$ works, the attacker only needs
        $m-1-m\varepsilon>m-2$ quantum shards, so adding this worst-case scenario gives
        \begin{equation}
            \eta_1=    \binom k {m-2}
            p^{m-2}
            (1-p)^{k-m+2}.
        \end{equation}
        Next, the attacker must combine these quantum shards into a Quantum
        bitcoin before another $T_{max}$ runs out. The probability for this is
        the probability of mining a single block:
        \begin{equation}
            \eta_2=
            \binom k 1
            p
            (1-p)^{k-1}
            =
            k p
            (1-p)^{k-1}.
        \end{equation}
        The total probability of reuse attack $\eta$ is then
        \begin{equation}
            \eta:=\eta_1\eta_2=
            \binom k {m-2}
            k
            \left( \frac  {p}{1-p} \right) ^{m-1}
            (1-p)^{2k}.
        \end{equation}
        We bound the binomial coefficient by above using the formula
        \begin{equation}
            \binom n k
            <
            \left( \frac {ne}k\right) ^k \text{ for } 1\leq k\leq n,
        \end{equation}
        which gives
        \begin{equation}
            \eta <
            \left( \frac {ke}{m-2}\right) ^{m-2}
            k
            \left( \frac  {p}{1-p} \right) ^{m-1}
            (1-p)^{2k}
            \text{ for } 2\leq m\leq k+1.
        \end{equation}
        We set $m-2=\gamma k$ which gives
        $1/k<\gamma<\min(1,1/k\varepsilon - 1/k)$. Such $\gamma$ exist since $\varepsilon<1/2$,
        which is the case since the Quantum Bitcoin system $\mathcal Q$ is
        secure. We then have
        \begin{equation}
            \eta <
            k
            \left( \frac e\gamma \cdot \frac  {p}{1-p} \right) ^{\gamma k}
            \left( \frac  {p}{1-p} \right)
            (1-p)^{2k},
            \label{eqn:eta}
        \end{equation}
        and note that
        \begin{equation}
            \frac e\gamma \cdot \frac  {p}{1-p} < \frac12 \Leftrightarrow
            0\leq p<\frac{\gamma}{2e+\gamma},
            \label{eqn:p-from-gamma}
        \end{equation}
        where the upper bound of $p$ approaches $1/(2e+1)\approx \SI{15.5}{\percent}$ as
        $\gamma$ goes to 1, under the condition that $k\leq 1/\varepsilon -1$. If
        $k$ is greater than $1/\varepsilon-1$, the upper bound on $p$ is even
        lower. However, $\varepsilon$ is exponentially small in $n$, so we can
        expect the \SI{15.5}{\percent} bound to be the correct one, given large
        enough $n$. Under the above constraints we get $\sup p/(1-p)=1/2e$ and
        $(1-p)^{2k}\leq 1$. Plugging in all this in Equation\nobreakspace \textup {(\ref {eqn:eta})} we get the
        following strict upper bound for the reuse attack probability:
        \begin{equation}
            \eta \ <
            \frac k {2e}
            2^{-\gamma k}.
            \label{eqn:reuse-attack}
        \end{equation}
    \end{proof}
    In other words, Quantum Bitcoin is secure against reuse attack as long as
    the attacker controls less than $\SI{15}{\percent}$ of the computing power.
    Note that Equation\nobreakspace \textup {(\ref {eqn:reuse-attack})} is the worst-case approximation and we
    should expect a much lower attack probability in a real scenario. What
    remains is to determine the parameter $\gamma$ introduced in the above
    proof. Too large, and it will be difficult for \emph{any} quantum bitcoin to
    be mined as every single quantum shard must be sold to a Quantum Bitcoin
    miner before $T_{max}$ runs out. Too small, and it becomes easier for a
    malicious miner to perform the reuse attack. The smaller we make $\gamma$,
    the larger $k$ must be in order to achieve the required bound on the attack
    probability.

    \subsection{Quantum Bitcoin Longevity}\label{sec:longevity}
    What remains is to show that a quantum bitcoin does not wear out too
    quickly, i.e. that they can be verified enough number of times to be usable.
    To prove this, we will use the following lemma due to
    \textcite{Aaronson2004}:
    \begin{lemma}[Almost as Good as New]
        Suppose a measurement on a mixed state $\rho$ yields a particular outcome
        with probability $1-\varepsilon$. Then after the measurement, one can
        recover a state $\widetilde \rho$ such that $\left\Vert \widetilde
        \rho - \rho\right\Vert_{tr} \leq \sqrt \varepsilon$
        \label{lem:good-as-new}
    \end{lemma}
    We now state the main longevity theorem:
    \begin{theorem}[Quantum Bitcoin Longevity]\label{thm:longevity}
        The number of times a quantum bitcoin can be verified and reconstructed
        is exponentially large in $n$.
    \end{theorem}
    \begin{proof}
        Corollary\nobreakspace \ref {thm:quantum-bitcoin} shows that the completeness error
        $\varepsilon$ of $\mathcal Q$ is exponentially small in $n$.
        When verifying a genuine quantum bitcoin $\$$, the
        verifier performs the measurement $\mathsf{Verify}_{\mathcal
        Q}(\$)$ on the underlying quantum states $\rho$, which
        yields the outcome \enquote{Pass} with probability
        $1-\varepsilon$. Then lemma\nobreakspace \ref {lem:good-as-new} shows that we can
        recover the underlying quantum states $\widetilde \rho_i$ of
        $\$$ so that $\left\Vert \widetilde \rho_i -
        \rho_i \right\Vert_{tr}\leq \sqrt{\varepsilon}$. As
        $\varepsilon$ is exponentially small in $n$, the trace distance
        becomes exponentially small in $n$ as well.
        Each time such a quantum bitcoin is verified and
        reconstructed, the trace distance between the \enquote{before} and
        \enquote{after} is exponentially small in $n$. Given any threshold after
        which we consider the quantum bitcoin \enquote{worn out}, the number of
        verifications it survives before passing this threshold is exponential in
        $n$.
    \end{proof}
    Theorem\nobreakspace \ref {thm:longevity} shows that a quantum bitcoin $\$$ can be verified and re-used many
    times before the quantum state is lost (assuming the absence of noise
    and decoherence). This is of course analogous to traditional, physical
    banknotes and coins which are expected to last for a large enough number
    of transactions before wearing out.
\end{appendices}

\printbibliography
\end{document}